\documentclass[11pt]{amsart}

\usepackage[colorlinks=true, pdfstartview=FitV, linkcolor=blue,
citecolor=blue]{hyperref}

\usepackage{a4wide}
\usepackage{amssymb,amsmath,amscd}

\DeclareFontFamily{U}{mathx}{\hyphenchar\font45}
\DeclareFontShape{U}{mathx}{m}{n}{ <5> <6> <7> <8> <9> <10>
   <10.95> <12> <14.4> <17.28> <20.74> <24.88> mathx10 }{}
\DeclareSymbolFont{mathx}{U}{mathx}{m}{n}
\DeclareFontSubstitution{U}{mathx}{m}{n}
\DeclareMathAccent{\widecheck}{0}{mathx}{"71}

\theoremstyle{plain}
\newtheorem{theorem}{Theorem}[section]

\newtheorem{lemma}[theorem]{Lemma}
\newtheorem{coro}[theorem]{Corollary}

\theoremstyle{definition}
\newtheorem{definition}[theorem]{Definition}
\newtheorem{example}[theorem]{Example}
\newtheorem{remark}[theorem]{Remark}

\newcommand{\ts}{\hspace{0.5pt}}
\newcommand{\nts}{\hspace{-0.5pt}}

\newcommand{\CC}{\mathbb{C}\ts}
\newcommand{\RR}{\mathbb{R}\ts}
\newcommand{\ZZ}{{\ts \mathbb{Z}}}
\newcommand{\QQ}{{\ts \mathbb{Q}}}

\newcommand{\SSS}{\mathbb{S}}

\newcommand{\NN}{\mathbb{N}}

\newcommand{\cA}{\mathcal{A}}
\newcommand{\cB}{\mathcal{B}}

\newcommand{\cM}{\mathcal{M}}
\newcommand{\cL}{\mathcal{L}}

\newcommand{\vL}{\varLambda}
\newcommand{\ii}{\mathrm{i}\ts}
\newcommand{\ee}{\mathrm{e}}
\newcommand{\eps}{\varepsilon}
\newcommand{\dd}{\, \mathrm{d}}

\newcommand{\oplam}{\mbox{\Large $\curlywedge$}}
\newcommand{\smoplam}{\mbox{\small $\curlywedge$}}
\newcommand{\exend}{\hfill $\Diamond$}
\newcommand{\defeq}{\mathrel{\mathop:}=}

\DeclareMathOperator{\dens}{dens}

\DeclareMathOperator{\vol}{vol}

\DeclareMathOperator{\supp}{supp}

\newcommand{\Cu}{C_{\mathsf{u}}}
\newcommand{\Cc}{C_{\mathsf{c}}}

\newcommand{\WAP}{\mathcal{W\nts A\ts P}}
\newcommand{\SAP}{\mathcal{S\nts A \ts P}}
\newcommand{\myfrac}[2]{\frac{\raisebox{-2pt}{$#1$}}
      {\raisebox{0.5pt}{$#2$}}}

\begin{document}

\title{Eberlein decomposition for PV inflation systems}

\author{Michael Baake}
\address{Fakult\"at f\"ur Mathematik,
         Universit\"at Bielefeld, \newline
\hspace*{\parindent}Postfach 100131, 33501 Bielefeld, Germany}
\email{mbaake@math.uni-bielefeld.de }

\author{Nicolae Strungaru}
\address{Department of Mathematical Sciences,
         MacEwan University, \newline
\hspace*{\parindent}10700 \ts 104 Avenue,
         Edmonton, AB, Canada T5J 4S2}
\email{strungarun@macewan.ca}

\begin{abstract}
  The Dirac combs of primitive Pisot--Vijayaraghavan (PV) inflations
  on the real line or, more generally, in $\RR^d$ are analysed.  We
  construct a mean-orthogonal splitting for such Dirac combs that
  leads to the classic Eberlein decomposition on the level of the pair
  correlation measures, and thus to the separation of pure point
  versus continuous spectral components in the corresponding
  diffraction measures. This is illustrated with two guiding examples,
  and an extension to more general systems with randomness is
  outlined.
\end{abstract}

\maketitle

\section{Introduction}

Symbolic Pisot--Vijayaraghavan (PV) substitutions induce a
much-studied class of dynamical systems under the action of $\ts\ZZ$.
By means of suitable suspensions, they also define natural dynamical
systems under the continuous translation action of $\ts\RR$.  Of
particular interest is the \emph{self-similar} suspension, which turns
the symbolic substitution system into a tiling inflation; see
\cite[Ch.~4]{TAO1} and references therein, as well as \cite{KLS,TAO2},
for general background. This setting is naturally connected with
general inflation tilings in $\RR^d$, which is our point of view here.

As the famous Pisot (or PV) substitution conjecture is still
unresolved, despite great effort and progress (see \cite{Aki} for a
summary), it seems a good strategy to consider such systems in a wider
setting, where one particularly takes mixed spectra more into
focus. So, given a general PV inflation system, it is of considerable
interest to decompose its spectrum in a constructive fashion. Here, we
report on some progress in this direction, where we start from the
Dirac comb of a PV inflation and split it into two parts, one of which
leads to the pure point part of the diffraction measure and the other
to the continuous part.  Moreover, this splitting possesses an
orthogonality relation in an averaged (or Eberlein) sense, which can
also be established for more general systems.

An important predecessor of our approach is the work by Aujogue
\cite{Au}, where an Eberlein-type decomposition is established for
measure-theoretic dynamical systems, hence in an almost sure
sense. When dealing with the class of primitive inflation tilings (or
a characteristic point set representing them), which define strictly
ergodic Delone dynamical systems, one wants to achieve such a
decomposition constructively, starting from a single member of the
dynamical system, that is, from a single point set in $\RR^d$, say.
Below, in view of later extensions, we do not restrict our attention
to Delone sets, but allow for more general point sets; see
\cite{Ri,MR,NS-w} for some of the theory that will then become useful.
Starting from the Dirac comb of such a point set, we construct a
\emph{splitting} into two measures that results in the Eberlein
decomposition for the pair correlation measures, and do this in such a
way that these two measures are mutually orthogonal in an averaged (or
Eberlein) sense. For some systems, such a splitting has been used in
the treatment of diffraction theory of systems with mixed spectrum,
see \cite{TAO1,NS-w,NS11} and references therein, both for
deterministic and for stochastic systems. \smallskip

The paper is organised as follows. First, in Section~\ref{sec:prelim},
we recall some of the necessary tools and results, which is
systematically formulated for $\RR^d$ to accommodate self-similar
inflation tilings in sufficient generality. This is followed by a
guiding example for $d=1$ with mixed spectrum, namely a twisted
version of the classic Fibonacci tiling, in
Section~\ref{sec:Fibo}. Then, we state and prove the central
orthogonality result, which takes most of Section~\ref{sec:ortho},
before we can formulate the decomposition theorem in
Section~\ref{sec:decomp}.  Here, we also show how it works for the
Thue{\ts}--Morse system, which has an inflation factor that is not a
unit and serves as our second guiding example.  Finally, in
Section~\ref{sec:further}, we extend our splitting approach to two
systems of stochastic nature, namely the interaction-free lattice gas
and the Fibonacci random inflation system from \cite{GL}.

\section{Preliminaries}\label{sec:prelim}

Due to our setting with Euclidean inflation tilings, we work with
$\RR^d$. Let us mention in passing that many steps can be
  generalised to any second countable, locally compact Abelian group
  as well, but we suppress this in what follows.  Let
$\cA = (A^{}_{n})^{}_{n\in\NN}$ be an \emph{averaging sequence} in
$\RR^d$, by which we denote a sequence of compact sets that are
\emph{nested}, meaning $A^{}_{n} \subset A^{\circ}_{n+1}$ for all
$n\in\NN$, and \emph{exhausting}, which refers to
$\bigcup_n A^{}_{n} = \RR^d$.  We call an averaging sequence
\emph{symmetric} when $A_n = - A_n$ holds for all $n \in \NN$, and
write this as $\cA = - \cA$.  Later, we shall only consider averaging
sequences that have the \emph{van Hove} property, which is to say
that, for any compact set $K\subset \RR^d$, one has
\begin{equation}\label{eq:van-Hove}
  \lim_{n\to\infty} \frac{\vol (\partial^K \! A_n )}{\vol (A_n)}
  \, = \, 0 \ts ,
\end{equation}
where
$\partial^K C \defeq \bigl( (C+K)\setminus C^{\circ} \bigr) \cup
\bigl( ( \overline{\RR^d\setminus C} - K )\cap C \bigr)$ for $K$ and
$C$ compact; compare \cite[p.~29]{TAO1} and references given there for
more. Symmetric van Hove averaging sequences in $\RR^d$ that are
widely used include cubes and balls, such as
$\bigl( [-n,n]^d \bigr)_{n\in\NN}$ or
$\bigl( \{ \| x \| \leqslant n \} \bigr)_{n\in\NN}$.

Recall that $\Cu (\RR^d)$ denotes the complex vector space of bounded
functions on $\RR^d$ that are uniformly continuous.

\begin{definition}\label{def:Eberlein}
  Two functions $f,g \in \Cu (\RR^d)$ are said to possess
  an \emph{Eberlein convolution} with respect to,
  or along, a given averaging sequence $\cA$ in $\RR^d$ if
\[
     \bigl( f \stackrel{\cA}{\circledast} g \bigr) (x) \, \defeq
     \lim_{n\to\infty} \myfrac{1}{\vol (A_n) } \int_{\nts A^{}_n} \!
     f(x-t) \, g(t) \dd t
\]
exists for all $x \in \RR^d$. When $\cA$ is a fixed averaging sequence
that has the van Hove property and is symmetric, so $\cA = - \cA$, we
will usually write $f\circledast g$ instead of
$f \stackrel{\cA}{\circledast} g$.
\end{definition}

At this point, it is relevant to ask whether or when $\circledast$ is
commutative.

\begin{lemma}
  Let\/ $\cA$ be a van Hove averaging sequence in\/ $\RR^d$, let\/
  $f,g \in \Cu (\RR^d)$, and assume that the Eberlein convolution of\/
  $f$ and\/ $g$ exists along\/ $\cA$. Then, also the Eberlein
  convolution of\/ $g$ and\/ $f$ exists, this time along\/ $-\cA$, and
  one has
\[
     g \nts \stackrel{- \cA}{\circledast} \nts f \, = \,
       f \stackrel{\cA}{\circledast} \ts g \ts .
\]
In particular, if\/ $\cA$ is also symmetric, one has\/
$f \circledast \ts g = g \ts \circledast f$.
\end{lemma}

\begin{proof}
  This follows from the following (backwards) calculation,
\begin{align*}
   \bigl( f \stackrel{\cA}{\circledast} \ts g \bigr) (x) \, & =
   \lim_{n\to\infty} \myfrac{1}{\vol (A_n)} \int_{A_n} \!
     f(x-s) \, g(s) \dd s \\[2mm]
   & =  \lim_{n\to\infty} \myfrac{1}{\vol (A_n)} \int_{x\ts -A_n}
      \! f(r) \, g(x-r) \dd r \\[2mm]
   & =  \lim_{n\to\infty} \myfrac{1}{\vol (A_n)} \int_{-A_n}
      \! g (x-r) \, f(r) \dd r \, = \,
    \bigl( g \nts \stackrel{-\cA}{\circledast} \nts f \bigr) (x)  \ts ,
\end{align*}
where the first equality in the last line is a consequence of the van
Hove property of $\cA$ together with the boundedness of $f$ and $g$.
\end{proof}

From now on, whenever we write $f \circledast g$, it is understood
that the Eberlein convolution refers to a symmetric van Hove averaging
sequence and is assumed to exist, that is, the averaging limit along
$\cA$ exists for all $x \in \RR^d$. Here, we need a generalisation of
this notion to translation-bounded Radon measures on $\RR^d$, which
are the measures $\mu$ such that
$\sup_{t\in\RR^d} | \mu | (t +K) < \infty$ for some fixed compact set
$K\subset \RR^d$ with non-empty interior, where $| \mu |$ denotes the
total variation of $\mu$. We denote the class of translation-bounded
Radon measures by $\cM^{\infty} (\RR^d)$, and say that
$\mu, \nu \in \cM^{\infty} (\RR^d)$ possess an Eberlein convolution
with respect to a symmetric van Hove averaging sequence $\cA$ if the
limit
\[
    \mu \circledast \nu \, =  \lim_{n\to\infty}
    \frac{ \mu |^{}_{A_n} \! * \nu |^{}_{A_n}}{\vol (A_n)}
\]
exists in the vague topology, where $\mu |^{}_{K}$ denotes the
restriction of the measure $\mu$ to a compact set $K \subset \RR^d$;
see \cite{TAO1,MoSt} for background and \cite{LSS} for further
details. Note that this definition is the unique extension of the
concept for functions to the setting of Radon measures.

\begin{remark}\label{rem:commute}
  The restriction to symmetric van Hove averaging sequences looks a
  little artificial. It was chosen to bypass a small inconsistency in
  the standard definition from above, which is widely used in
  the literature. Indeed, for a general van Hove averaging sequence,
  a more consistent alternative would be
\[
    \mu \stackrel{\cA}{\circledast} \nu \, \defeq \lim_{n\to\infty}
    \frac{ \mu |^{}_{-A_n} \! * \nu |^{}_{A_n}}{\vol (A_n)} \ts ,
\]
which satisfies
$\mu \stackrel{\cA}{\circledast} \nu = \nu
\!\stackrel{-\cA}{\circledast}\! \mu$ in analogy to above. This
definition is the consistent extension of
Definition~\ref{def:Eberlein}.  Indeed, when $\mu$ and $\nu$ are
absolutely continuous measures with Radon--Nikodym densities $f$ and
$g$, respectively, one finds
\[
    \mu \stackrel{\cA}{\circledast} \nu \, = \,
    \bigl( f \stackrel{\cA}{\circledast} g \bigr) \lambda^{}_{\mathrm{L}}
    \quad \text{and} \quad
    \nu \stackrel{\cA}{\circledast} \mu \, = \,
    \bigl( g \stackrel{\cA}{\circledast} f \bigr)
    \lambda^{}_{\mathrm{L}} \ts ,
\]
where $\lambda^{}_{\mathrm{L}}$ denotes Lebesgue measure on $\RR^d$.
This observation follows from
\[
\begin{split}
   \bigl( \mu |^{}_{-A_n} \! * \nu |^{}_{A_n} \bigr) (h) \, & = \,
   \int_{\RR^d} \int_{\RR^d} h(x+y) \, 1^{}_{-A^{}_{n}} (x) \, f(x) \,
      1^{}_{\nts A^{}_{n}} (y) \, g(y) \dd x \dd y \\[2mm]
      & = \, \int_{\RR^d} h(s) \int_{A_n \cap (s+A_n)} f(s-y) \, g(y)
       \dd y \dd s  \ts ,
\end{split}
\]
where $h \in \Cc (\RR^d)$ is arbitrary, after dividing by $\vol (A_n)$
and taking the limit as $n\to\infty$.

When $\cA$ is also symmetric, these subtleties go away and
$\circledast$ becomes commutative. Since this extra assumption poses
no relevant restriction to any of our later arguments, we will usually
make it, but always say so in our formal statements.  \exend
\end{remark}

Next, we need another notion, the \emph{Fourier--Bohr} (FB)
coefficients.

\begin{definition}\label{def:FB}
  A function $g \in \Cu (\RR^d)$ possesses the \emph{FB
    coefficient} at $k\in\RR^d$ with respect to $\cA$ if
\[
      c^{}_{g} (k) \, \defeq \lim_{n \to\infty}
      \myfrac{1}{\vol (A_n)} \int_{A^{}_{n}} \!
      \ee^{-2\pi\ii k x} g (x)  \dd x
\]
exists. We say that $g$ possesses FB coefficients relative to $\cA$
when $c^{}_{g} (k)$ exists for all $k \in \RR^d$.

Likewise, a Radon measure $\mu \in \cM^\infty (\RR^d)$
possesses FB coefficients with respect to $\cA$ on a set
$S\subseteq \RR^d$ if
\[
    c^{}_{\mu} (k) \, \defeq \lim_{n \to \infty}
    \myfrac{1}{\vol (A_n)} \int_{A^{}_{n}} \! \ee^{-2\pi\ii kx} \dd \mu (x)
\]
exists for all $k\in S$.
\end{definition}

From now on, whenever we write $c^{}_{g} (k)$ or $c^{}_{\mu} (k)$, it
is understood that the corresponding limit exists.

\smallskip

The following result, which can be derived from \cite[Lemma~8]{Lenz}
and is a mild variant of \cite[Lemma~1.9 and Cor.~1.20]{LSS2}, see
also \cite[Prop.~8.2]{ARMA}, provides the connection between the FB
coefficients of a translation-bounded Radon measure $\mu$ and the
convolutions $\mu*\varphi$ with arbitrary $\varphi \in \Cc (\RR^d)$,
where the latter denotes the space of continuous functions on $\RR^d$
with compact support. Below, we use $\widehat{\varphi}$ for the
Fourier transform of a function $\varphi$, employing the conventions
of \cite[Ch.~8]{TAO1}.

\begin{lemma}\label{lem:FB-rel}
  Let\/ $\cA$ be a general van Hove averaging sequence in\/ $\RR^d$,
  and consider a Radon measure\/ $\mu \in \cM^\infty (\RR^d)$. If\/
  $\varphi \in \Cc (\RR^d)$ and\/ $k \in \RR^d$, one has
\[
   \lim_{n \to \infty} \myfrac{1}{\vol (A_n)} \,
    \biggl| \int_{A^{}_n} \! \ee^{-2\pi\ii k t}
    \bigl(\varphi*\mu \bigr) (t) \dd t  \, - \,
    \widehat{\varphi} \ts (k) \! \int_{A^{}_n} \!
    \ee^{-2\pi\ii k t} \dd \mu(t)    \biggr| \,  = \,  0 \ts .
\]
In particular, if\/ $c^{}_{\mu} (k)$ exists, then so does\/
$c^{}_{\varphi*\mu} (k)$, and one has
\[
     c^{}_{\varphi \ts *\mu} (k) \, = \, \widehat{\varphi} \ts (k)\,
     c^{}_{\mu} (k ) \ts .
\]
Conversely, if\/ $c^{}_{\varphi \ts *\mu} (k)$ exists for some\/
$\varphi$ with\/ $\widehat{\varphi}\ts (k) \ne 0$, then\/
$c^{}_{\mu} (k)$ exists as well.
\end{lemma}

\begin{proof}
Let us first note that
\[
\begin{split}
   \widehat{\varphi}(k) \int_{A^{}_n} \! \ee^{-2 \pi \ii k t} \dd \mu(t)
   \, & = \, \widehat{\varphi}(k) \int_{\RR^d}
    1^{}_{\nts A_n}(s) \, \ee^{-2 \pi \ii k s} \dd \mu(s)  \\[2mm]
  & =  \int_{\RR^d} \int_{\RR^d} \ee^{-2 \pi \ii k (r+s)} \ts \varphi (r)
     \dd r \; 1^{}_{\nts A_n}(s) \dd \mu(s)  \\[2mm]
  &=  \int_{\RR^d} \int_{\RR^d} \! 1^{}_{\nts A_n}(s) \, \ee^{-2 \pi \ii kt}
    \ts \varphi (t-s) \dd \mu(s) \dd t \ts ,
\end{split}
\]
where we used the substitution $t = r+s$ and Fubini's theorem for the
last line.

Consequently, we have
\[
\begin{split}
  \bigg| \int_{A^{}_n} \! \ee^{-2 \pi \ii k t} &
      \bigl( \varphi*\mu \bigr) (t) \dd t \, - \, \widehat{\varphi}(k)
       \int_{A^{}_n} \! \ee^{-2 \pi \ii k t} \dd \mu(t) \ts \bigg|  \\[2mm]
  &= \, \bigg| \int_{\RR^d} \int_{\RR^d}  \ee^{-2 \pi \ii k t} \ts
      \bigl( 1^{}_{\nts A_n}(t) - 1^{}_{\nts A_n}(s) \bigr)
      \, \varphi(t-s) \dd \mu(s) \dd t \ts \bigg| \\[2mm]
   &\leqslant \, \int_{\RR^d} \int_{\RR^d} \big|
      1^{}_{\nts A_n}(t) - 1^{}_{\nts A_n}(s)\big|  \, | \varphi(t-s) |
      \dd \lvert \mu \rvert (s) \dd t  .
\end{split}
\]
With $K =\supp(\varphi)$, we have
$\big| 1^{}_{\nts A_n}(t) - 1^{}_{\nts A_n}(s) \big| \, \lvert
\varphi(t-s)\rvert =0$ for any $t \notin \partial^K \! A_n$, hence
\[
    \big| 1^{}_{\nts A_n}(t) - 1^{}_{\nts A_n}(s)\big| \,
    | \varphi(t-s)| \, \leqslant \,
    1^{}_{\partial^K \! A_n}(t) \, \lvert \varphi (t-s) \vert \ts ,
\]
which then gives
\[
\begin{split}
   \myfrac{1}{\vol(A_n)} & \bigg| \int_{A^{}_n} \! \ee^{-2 \pi \ii k t}
      \bigl( \varphi*\mu \bigr) (t) \dd t \, - \,
      \widehat{\varphi} (k) \int_{A^{}_n} \!  \ee^{-2 \pi \ii k t}
      \dd \mu (t) \ts \bigg| \\[2mm]
  & \leqslant \, \myfrac{1}{\vol(A_n)}  \int_{\RR^d} \int_{\RR^d}
       1^{}_{\partial^K \! A_n} (t) \,
       \lvert \varphi (t-s) \rvert
       \dd \lvert \mu \rvert (s) \dd t \\[2mm]
   & = \, \myfrac{1}{\vol(A_n)}  \int_{\partial^K \! A_n}
        \int_{\RR^d} \vert \varphi (t-s) \rvert \dd \vert \mu \rvert (s)
        \dd t \\[2mm]
   &= \, \myfrac{1}{\vol(A_n) } \int_{\partial^K \! A_n} \!
       \bigl( \lvert \varphi \rvert * \lvert \mu \rvert \bigr) (t)
       \dd t  \:  \leqslant  \: \big\| \lvert \varphi \rvert *
       \vert \mu \rvert \big\|_{\infty}
       \frac{ \vol (\partial^K A_n )}{\vol(A_n)} \ts .
\end{split}
\]
As $n\to\infty$, the claim now follows from the
translation-boundedness of $\mu$ and the van Hove property
of $\cA$ from Eq.~\eqref{eq:van-Hove}.
\end{proof}

What we wrote down so far are special relations involving continuous
\emph{characters} on $\RR^d$, which are the elements of the dual
group. Since $\RR^d$ is self-dual, it is common to employ an additive
notation for the characters, and denote them by
$\chi^{}_{k} \colon \RR^d \xrightarrow{\quad} \SSS^1$, where
$k\in\RR^d$ is fixed and the mapping is given by
\[
    x \, \mapsto \, \chi^{}_{k} (x)  \defeq  \ee^{2 \pi \ii k x} .
\]
Clearly, one then has $\chi^{}_{k} (x) \ne 0$ for all $x\in\RR^d$,
together with $\overline{\chi^{}_{k}} = \chi^{}_{-k}$.

Before we embark on the general result, let us discuss our first
guiding example for $d=1$, which is built as a simple
extension \cite{BFG} of the classic Fibonacci tiling of the real line.

\section{A simple PV inflation with mixed spectrum}\label{sec:Fibo}

Fix the alphabet $\{ a, \underline{a}, b, \underline{b}\ts  \}$
and consider the aperiodic substitution rule
\[
    \varrho \colon \;  a \mapsto a\ts b \ts , \quad
    \underline{a} \mapsto \underline{a\nts}\ts\ts \underline{b} \ts , \quad
    b \mapsto \underline{a} \ts , \quad
    \underline{b} \mapsto a \ts ,
\]
which has the substitution matrix
\[
     M^{}_{\varrho} \, = \, \begin{pmatrix}
     1 & 0 & 0 & 1 \\ 0 & 1 & 1 & 0 \\
     1 & 0 & 0 & 0 \\ 0 & 1 & 0 & 0 \end{pmatrix} .
\]
Since $M^{}_{\varrho}$ has Perron--Frobenius eigenvalue
$\tau = \frac{1}{2} ( 1 + \sqrt{5}\, )$, which is a PV unit, with
corresponding left eigenvector $(\tau, \tau, 1, 1 )$, one can turn
$\varrho$ into a PV inflation with prototiles (intervals) of length
$\tau$ for $a$ and $\underline{a}\ts$, and of length $1$ for $b$ and
$\underline{b}\ts$; see \cite[Sec.~3.2]{BFG} and references therein
for details. On identifying $a$ with $\underline{a}$ and $b$ with
$\underline{b}\ts$, one obtains the classic Fibonacci tilings of the
real line \cite[Ex.~4.6 and Sec.~9.4.1]{TAO1} with pure point
spectrum, both in the diffraction and in the dynamical sense. This
implies that the `twisted' inflation system induced by $\varrho$ is
an almost everywhere $2:1$ extension of the classic Fibonacci
system. By standard results, see \cite{ABKL} and references
  therein, this implies that the mapping of the twisted system onto
its maximal equicontinuous factor is also $2:1$ almost
everywhere. Consequently, the twisted system must have mixed
spectrum, with a pure point and a continuous part.  The latter is
purely singular continuous in this case, by
\cite[Thm.~3.2]{BFG}.  Let us sketch how to arrive at this conclusion
constructively.

Now, working with an inflation fixed point and the left endpoints of
the intervals of type
$\alpha \in \{ a, \underline{a}\ts , b, \underline{b}\ts \}$, we
always get $\vL_{\alpha} \in \ZZ[\tau]$, wherefore we can employ the
natural \emph{cut and project scheme} (CPS) of the Fibonacci system,
abbreviated by $(\RR, \RR, \cL)$, as constructed and described in
detail in \cite[Sec.~7.2]{TAO1}. Here, the CPS is given by
\[
   \renewcommand{\arraystretch}{1.2}\begin{array}{r@{}ccccc@{}l}
   & \RR & \xleftarrow{\;\;\; \pi \;\;\; }
      & \RR \nts\nts \times \nts\RR &
        \xrightarrow{\;\: \pi^{}_{\text{int}} \;\: } & \RR & \\
   & \cup & & \cup & & \cup & \hspace*{-1ex}
   \raisebox{1pt}{\text{\footnotesize dense}} \\
   & \pi (\cL) & \xleftarrow{\;\ts 1-1 \;\ts } & \cL &
        \xrightarrow{ \qquad } &\pi^{}_{\text{int}} (\cL) & \\
   & \| & & & & \| & \\
   & \ZZ[\tau] & \multicolumn{3}{c}{\xrightarrow{\qquad\qquad\:\,\star
       \,\:\qquad\qquad}}
       &  {\ZZ[\tau]}  & \\
\end{array}\renewcommand{\arraystretch}{1}
\]
where $\star$ denotes the star map of the CPS, which is the
restriction to $\ZZ(\tau)$ of the unique field automorphism of
$\QQ (\sqrt{5}\,)$ induced by $\sqrt{5} \mapsto - \sqrt{5}$, and
$\cL \defeq \big\{ (x, x^{\star}) : x \in \ZZ[\tau] \big\}$ is the
Minkowski embedding of $\ZZ[\tau]$, which is a lattice in
$\RR^2\simeq \RR \times \RR$ of density $1/\sqrt{5}\ts$; see
\cite[Sec.~3.4 and Ex.~7.3]{TAO1} for background and details. Setting
$W^{}_{\!\alpha} = \overline{\vL^{\star}_{\alpha}}$, one obtains
\cite{BFG} the compact intervals
\[
    W_{\! a} \, = \, W_{\! \underline{a}} \, = \, [\tau-2,\tau-1]
    \quad \text{and} \quad
    W_{\nts b} \, = \, W_{\nts \underline{b}} \, = \, [-1,\tau-2] \ts .
\]
Now, each set
$\oplam (W_{\nts \alpha}) \defeq \{ x \in \ZZ[\tau]: x^{\star} \in
W_{\nts \alpha} \}$ is a regular model set with pure point diffraction
\cite[Thm.~9.4]{TAO1}, while this is \emph{not} true of the point sets
$\vL_{\alpha}$.

One can check that each $\vL_{\alpha}$, which satisfies
$\vL_{\alpha} \subseteq \oplam (W_{\nts \alpha})$ by construction, has
half the density of $\oplam (W_{\nts \alpha})$; compare
\cite[Sec.~9.4.1]{TAO1} for the results of the two-letter Fibonacci
chain.  By \cite[Thm.~5.3]{BG}, the point sets $\vL^{\star}_{\alpha}$
are uniformly distributed in $W^{}_{\! \alpha}$, and each of the
measures
\[
    \nu^{}_{\alpha} \, = \, \delta^{}_{\! \vL_{\alpha}}
       - \, \myfrac{1}{2} \, \delta^{}_{\nts \smoplam (W_{\alpha})}
\]
has vanishing FB coefficients, where we use $\delta^{}_{S} \defeq
\sum_{x\in S} \delta^{}_{x}$ to denote the \emph{Dirac comb} of a
point set $S$. So, with $\omega^{}_{\alpha} =
\frac{1}{2} \ts \delta^{}_{\nts \smoplam (W_{\alpha})}$, one has a
decomposition
\[
      \delta^{}_{\! \vL_{\alpha}} \, = \,
      \omega^{}_{\alpha} + \ts \nu^{}_{\alpha}
\]
into two summands. The crucial observation now is that the first gives
rise to the pure point part of the spectrum, and the second to the
continuous part, while all cross terms under the Eberlein convolution
with respect to any symmetric van Hove averaging sequence in $\RR$
vanish. Further details, and a closely related example with a more
complicated window structure, are discussed in \cite[Sec.~7]{BG}. Our
goal now is to substantiate the decomposition claim and prove it in
sufficient generality.

\section{Orthogonality for Eberlein convolution}\label{sec:ortho}

Let us first state an elementary result on the connection between the FB
coefficients of a function $f\in\Cu (\RR^d)$ and the Eberlein
convolution of $f$ with the characters on $\RR^d$.

\begin{lemma}\label{lem:FB-Eber}
  Let\/ $\cA$ be a symmetric van Hove averaging sequence in\/ $\RR^d$,
  and consider a function\/ $f \in \Cu (\RR^d)$ together with a
  character\/ $\chi^{}_{k}$. Then, with respect to\/ $\cA$, the
  following statements are equivalent.
\begin{enumerate}\itemsep=2pt
\item The FB coefficient\/ $c^{}_{f} (k)$ exists.
\item The Eberlein convolution\/ $\chi^{}_{k} \circledast f$ exists.
\item There is some\/ $x\in\RR^d$ such that\/
  $\lim_{n\to\infty} \frac{1}{\vol (A_n)} \int_{\RR^d} \chi^{}_{k}
  (x-t) \, f(t) \dd t$ exists.
\end{enumerate}
Further, assume that\/ $c^{}_{f} (k)$ exists. Then, for all\/
$x\in\RR^d$, one has the relation
\[
    \bigl( \chi^{}_{k}  \circledast f \bigr) (x) \, = \,
     \chi^{}_{k} (x) \, c^{}_{f} (k) \ts .
\]
\end{lemma}

\begin{proof}
For arbitrary $x \in \RR^d$ and $n \in \NN$, we have
\[
    \int_{A^{}_n} \! \chi^{}_{k} (x-t) \, f(t) \dd t \, =
    \int_{A^{}_n} \! \chi^{}_{k}(x) \,
    \overline{\chi^{}_{k}(t)} \, f(t) \dd t
    \, = \,  \chi^{}_{k} (x) \! \int_{A^{}_n} \!
   \overline{\chi^{}_{k}(t)} \, f(t) \dd t  \ts .
\]
Since $\chi^{}_{k} (x) \ne 0$, the claimed equivalences, as well as
the final identity, follow via dividing by $\vol (A_n)$ and taking the
limit as $n\to\infty$.
\end{proof}

This has an immediate consequence as follows.

\begin{coro}\label{coro:P-coeff}
  Let\/ $\cA$ be a symmetric van Hove averaging sequence in\/
  $\RR^d$, and consider a function\/ $f \in \Cu (\RR^d)$ together
  with a trigonometric polynomial\/
  $P= \sum_{j=1}^n \alpha^{}_j \, \chi^{}_{k_j}$. If the
  FB coefficients\/ $c^{}_{f} (k_j) $ exist for all\/
  $1 \leqslant j \leqslant n$, the
  Eberlein convolution\/ $P \circledast f$ exists, too, with
\[
    \pushQED{\qed}
    P \circledast f \, = \sum_{j=1}^{n}
    \alpha^{}_{j}  \, c^{}_{f} (k_j) \, \chi^{}_{k_j}  .
    \qedhere\popQED
\]
\end{coro}

To continue, we need another notion as follows.

\begin{definition}
  Let $\cA$ be a van Hove averaging sequence in $\RR^d$ and
  $\mu \in \cM^{\infty} (\RR^d)$.  If the FB coefficients of $\mu$
  with respect to $\cA$ exist, the corresponding \emph{FB spectrum}
  of $\mu$ is the set $\{ k \in \RR^d : c^{}_{\mu} (k) \ne 0 \}$.
   Further, we say that  $\mu$ has \emph{null FB spectrum}
   with respect to $\cA$ if, for all $k\in\RR^d$, one has
\[
      c^{}_{\mu} (k) \,  = \lim_{n \to\infty} \myfrac{1}{\vol (A_n)}
      \int_{A^{}_n} \! \overline{\chi^{}_{k} (t)} \dd \mu(t)
      \, = \, 0 \ts ,
\]
  which is to say that all FB coefficients of\/ $\mu$ exist and vanish.

  Likewise, $f \in \Cu (\RR^d)$ has \emph{null FB spectrum} with
  respect to $\cA$ if, for all $k\in\RR^d$, we have
\[
    c^{}_{f} (k) \, = \lim_{n \to\infty} \myfrac{1}{\vol (A_n)}
    \int_{A^{}_n} \! \overline{\chi^{}_{k} (t)} \, f(t)  \dd t
    \, = \, 0 \ts .
\]
\end{definition}

An immediate consequence of Lemma~\ref{lem:FB-Eber} and
Lemma~\ref{lem:FB-rel} is the following.

\begin{coro}\label{cor:FB-null}
  Let\/ $\cA$ be a symmetric van Hove averaging sequence in\/
  $\RR^d \nts$. Consider a measure\/ $\mu \in \cM^{\infty} (\RR^d )$
  and a function\/ $f\in \Cu (\RR^d)\nts$.  Then, the following two
  properties hold.
\begin{itemize}\itemsep=2pt
\item [(a)] If\/ $\mu$ has null FB spectrum with respect to $\cA$, the
  function\/ $\varphi * \mu$, for any\/ $\varphi \in \Cc (\RR^d)$, has
  null FB spectrum with respect to\/ $\cA$ as well.
\item [(b)] If\/ $f$ has null FB spectrum with respect to\/ $\cA$, one
  has\/ $P \circledast f = 0$ for all trigonometric polynomials\/ $P$,
  where\/ $\circledast$ is the Eberlein convolution along\/ $\cA$.
  \qed
\end{itemize}
\end{coro}

\medskip

At this point, we can prove the central tool for our later
computations, where $\delta^{}_{S}$ is again the {Dirac comb} of a
point set $S\subset \RR^d$. For generalisations of this
result, we refer to \cite{S}.

\begin{theorem}\label{thm:FB-trig}
  Let\/ $\cA$ be a symmetric van Hove averaging sequence in\/
  $\RR^d \nts$, consider a measure\/ $\mu \in \cM^{\infty} (\RR^d)$,
  and assume that\/ $\mu$ has null FB spectrum with respect to\/
  $\cA$. Further, let\/ $\vL$ be a regular model set in\/
  $\RR^d$. Then, the Eberlein convolution\/
  $\mu \circledast \delta^{}_{\!\vL}$ exists along\/ $\cA$ and
\[
     \mu \circledast \delta^{}_{\!\vL} \,  = \, 0 \ts .
\]
\end{theorem}

\begin{proof}
  Let $( \RR^d , H, \cL)$ be the CPS for the description of the model
  set, where $H$ is a compactly generated LCAG and $\cL$ is a lattice
  in $\RR^d \times H$, that is, a co-compact discrete subgroup. Note
  that we need $H$ in this generality because the internal space need
  not be Euclidean; we refer to \cite{Moo97,Bob-rev} and
  \cite[Sec.~7.2]{TAO1} for general background, and to \cite{BM,NS11}
  for a detailed description of the explicit construction of the CPS.
  Further, let $W\subseteq H$ be the window that gives
\[
    \vL \, = \, \oplam(W) \, = \,
     \{ x \in \pi (\cL) : x^{\star} \in W \}
\]
as a regular model set, where
$\pi \colon \RR^d \times H \xrightarrow{\quad} \RR^d$ is the canonical
projection and $\star$ denotes the star map of the CPS.

Now, since $\mu$ and $\delta^{}_{\!\vL}$ are translation-bounded
measures by construction, the set
\[
    \bigg\{ \frac{\mu |^{}_{A_n} \! * \delta^{}_{\!\vL}|^{}_{A_n} }
             {\vol (A_n)} : n \in \NN \bigg\}
\]
is pre-compact in the vague topology and metrisable \cite{BL}.
Therefore, to prove that $\mu \circledast \delta^{}_{\!\vL} =0$,
it suffices to show that $0$ is the only cluster point of this set
(or sequence).

Let $\eta$ be any cluster point of this sequence, and let
$\cB= ( B_n )^{}_{n\in\NN}$ be a subsequence of $\cA$ with respect
to which $\eta$ is a limit, which means that $B_n =A_{\ell_n}$ with
$\ell_{n+1} > \ell_n$ for all $n\in\NN$. Clearly, $\cB$ is again a
symmetric van Hove averaging sequence, and one has
\begin{equation}\label{eq:conv-approx}
   \eta \, = \lim_{n\to\infty} \frac{\mu |^{}_{B_n} \! *
        \delta^{}_{\!\vL}|^{}_{B_n}}{\vol (B_n)}
   \, = \lim_{n\to\infty} \frac{\mu |^{}_{B_n} \! *
        \delta^{}_{\!\vL}}{\vol (B_n)} \ts ,
\end{equation}
where the second equality follows from the van Hove property of $\cB$
via Schlottmann's lemma \cite[Lemma~1.2]{Martin}; see also
\cite{BL,LSS}.

Now, fix $\varphi, \psi \in \Cc (\RR^d)$, set $K:=\supp(\varphi)$, and
note that the translation-boundedness of $\mu$ implies
$\| \varphi * \mu \|^{}_\infty <\infty$.  Let $\eps >0$, and select a
function $h \in \Cc (H)$ that satisfies $1^{}_W \leqslant h$ together
with
\begin{equation}\label{eq:eps-1}
      \int_H \bigl( h(t)-1^{}_W (t)\bigr) \dd t \, < \,
      \frac{\eps}{1 + 2 \dens(\cL) \ts \| \varphi*\mu \|^{}_{\infty}
       \int_{\RR^d} \lvert \psi (t) \rvert \dd t } \ts ,
\end{equation}
which is clearly possible. Further, with
$L = \pi (\cL) \subset \RR^d$, let
\[
    \nu \, = \,  \omega^{}_{h} \, \defeq
    \sum_{ x \in L} h (x^\star) \, \delta_x  \ts ,
\]
which is a strongly almost periodic measure, $\nu \in \SAP(\RR^d)$, by
\cite[Thm.~5.5.2]{NS11}; see also \cite[Thm.~3.1]{LR}.  Consequently,
$\psi * \nu $ is a Bohr (or uniformly) almost periodic function. This
implies that there exists a trigonometric polynomial $P$ in $d$
variables such that
\begin{equation}\label{eq:P-def}
   \| \psi * \nu - P \|^{}_{\infty}  \, < \,
     \frac{\eps}{1 + 2 \ts \| \varphi*\mu \|^{}_{\infty}} \ts .
\end{equation}

In view of Eq.~\eqref{eq:conv-approx}, we now have
\begin{align*}
    \bigl( \varphi*\psi*\eta \bigr) (0) \, & =
    \lim_{n\to\infty} \myfrac{1}{\vol (B_n)} \int_{\RR^d}
        \bigl( \varphi*\psi \bigr) (-s)
       \dd \bigl( \mu |^{}_{B^{}_{n}}* \delta^{}_{\!\vL}\bigr) (s)  \\[2mm]
   &= \lim_{n\to\infty} \myfrac{1}{\vol (B_n)} \int_{\RR^d} \int_{\RR^d}
     \bigl( \varphi*\psi \bigr) (-x-y) \dd \mu |^{}_{B^{}_{n}}(x)
       \dd  \delta^{}_{\!\vL} (y)  \\[2mm]
   &= \lim_{n\to\infty} \myfrac{1}{\vol (B_n)} \int_{\RR^d}
       \int_{\RR^d} \int_{\RR^d}
      1^{}_{\nts B^{}_{n}}(x)\, \varphi(-x-y-s) \, \psi(s)
      \dd s \dd \mu(x) \dd  \delta^{}_{\!\vL} (y)  \ts ,
\intertext{which can be continued via the substitution $r = y+s$ and
    Fubini's theorem as}
     \bigl( \varphi*\psi*\eta \bigr) (0) \,
   & = \lim_{n\to\infty} \myfrac{1}{\vol (B_n)} \int_{\RR^d}
     \int_{\RR^d} \int_{\RR^d}
       1^{}_{\nts B^{}_{n}}(x)\, \varphi(-x-r)\, \psi(r-y) \dd r \dd \mu (x)
         \dd  \delta^{}_{\!\vL} (y)  \\[2mm]
   & = \lim_{n\to\infty} \myfrac{1}{\vol (B_n)} \int_{\RR^d}
        \int_{\RR^d} \int_{\RR^d}
     1^{}_{\nts B^{}_{n}}(x)\, \varphi(-x-r)\,  \psi(r-y) \dd
     \delta^{}_{\!\vL}(y)
          \dd r \dd \mu (x)   \\[2mm]
   & = \lim_{n\to\infty} \myfrac{1}{\vol (B_n)} \int_{\RR^d} \int_{\RR^d}
     1^{}_{\nts B^{}_{n}}(x) \, \varphi(-x-r)
     \bigl( \psi*\delta^{}_{\!\vL} \bigr) (r)
     \dd \mu (x) \dd r   \ts .
\end{align*}
Here, we observe that  $1^{}_{\nts B_n}(x) \ts \varphi(-x-r) =
     1^{}_{\nts B_n}(r) \ts \varphi(-x-r)$ holds for $x \not\in \partial^{K}
     \! B_n$. Due to the
van Hove property of $\cB$, we have $\lim_{n\to\infty} \vol(\partial^K \! B_n)
/\nts \vol (B_n) = 0$, so we get

\begin{align}
   \bigl( \varphi*\psi * \eta \bigr) (0) \, & =
     \lim_{n\to\infty} \myfrac{1}{\vol (B_n)} \int_{\RR^d}
      1^{}_{\nts B^{}_{n}}(r)\, \bigl( \psi*\delta^{}_{\!\vL}\bigr) (r)
      \int_{\RR^d}  \varphi(-x-r) \dd \mu(x)\dd r  \notag \\[2mm]
   & = \lim_{n\to\infty} \myfrac{1}{\vol (B_n)} \int_{\RR^d}
      1^{}_{\nts B^{}_{n}}(r)\, \bigl( \psi*\delta^{}_{\!\vL} \bigr) (r)
       \, \bigl( \varphi*\mu \bigr) (-r)   \dd r  \notag \\[2mm]
   & = \lim_{n\to\infty} \myfrac{1}{\vol (B_n)} \int_{B^{}_{n}}
     \bigl( \psi*\delta^{}_{\!\vL} \bigr) (r) \,
     \bigl( \varphi*\mu \bigr) (-r)
         \dd r \ts .  \notag
\intertext{Therefore, employing a standard $3\ts \eps$-strategy, we find}
   \big| \bigl( \varphi*\psi * \eta \bigr) (0) \big| \, &
  = \lim_{n\to\infty} \myfrac{1}{\vol (B_n)} \left| \int_{B^{}_{n}}
      \bigl( \psi*\delta^{}_{\!\vL} \bigr) (r) \,
      \bigl( \varphi*\mu \bigr) (-r)   \dd r  \ts \right| \nonumber \\[2mm]
  & \leqslant \, \limsup_{n\to\infty} \myfrac{1}{\vol (B_n)} \left|
      \int_{B^{}_{n}}  \bigl(\psi*\delta^{}_{\!\vL}-\psi*\nu \bigr) (r)
    \, \bigl( \varphi*\mu \bigr) (-r)   \dd r  \ts \right|
    \label{EQ1} \\[2mm]
  & \quad \, + \limsup_{n\to\infty} \myfrac{1}{\vol (B_n)}
       \left| \int_{B^{}_{n}}  \bigl(\psi*\nu-P \bigr) (r)
       \, \bigl( \varphi*\mu\bigr) (-r)   \dd r \ts  \right| \nonumber \\[2mm]
  & \quad \, + \limsup_{n\to\infty} \myfrac{1}{\vol (B_n)}
       \left| \int_{B^{}_{n}}  \! P (r) \, \bigl( \varphi*\mu \bigr) (-r)
        \dd r  \ts \right| ,  \nonumber
\end{align}
where $P$ is the trigonometric polynomial from \eqref{eq:P-def}.

Now, we need to estimate the three terms in the last expression of
\eqref{EQ1}, where we begin with the middle one. Here, for all
$n \in \NN$, we recall \eqref{eq:P-def} and obtain
\[
  \myfrac{1}{\vol (B_n)} \left|\int_{B^{}_{n}}
     \bigl( \psi*\nu-P \bigr) (r)\,  \bigl( \varphi*\mu \bigr) (-r)
       \dd r  \ts \right|
   \, \leqslant  \, \|  \psi*\nu-P \|^{}_{\infty} \,
       \| \varphi*\mu \|^{}_{\infty}  \, < \, \myfrac{\eps}{2}  \ts ,
\]
which then gives
\[
   \limsup_{n\to\infty} \myfrac{1}{\vol (B_n)} \left|
      \int_{B^{}_{n}}  \bigl( \psi*\nu - P \bigr) (r) \,
      \bigl(  \varphi*\mu \bigr) (-r)   \dd r \ts \right|
      \,  \leqslant \, \myfrac{\eps}{2}  \ts .
\]
Next, since $\mu$ has null FB spectrum relative to $\cA$ by
assumption, hence clearly also relative to the subsequence $\cB$,
Corollary~\ref{cor:FB-null} implies
\[
    \limsup_{n\to\infty} \myfrac{1}{\vol (B_n)}
    \left| \int_{B^{}_{n}}  P (r) \, \bigl(  \varphi*\mu \bigr) (-r)
    \dd r  \ts \right| \, = \, \big| P \stackrel{\cB}{\circledast}
     ( \varphi*\mu) \big| (0) \, = \,  0 \ts .
\]

The remaining term is a little harder. Here, we have
{\allowdisplaybreaks
\begin{align*}
    T^{}_{1} \, & \defeq \, \limsup_{n\to\infty} \,
       \myfrac{1}{\vol (B_n)} \left| \int_{B^{}_{n}}
       \bigl(\psi*\delta^{}_{\!\vL} \nts -\psi*\nu \bigr) (r)\,
       \bigl( \varphi*\mu \bigr) (-r)   \dd r  \ts \right|  \\[2mm]
  &\, \leqslant \, \limsup_{n\to\infty} \myfrac{1}{\vol (B_n)} \int_{B^{}_{n}}
      \nts \big| \psi*\delta^{}_{\!\vL} \nts -\psi*\nu \big|  (r) \,
       \| \varphi*\mu \|^{}_{\infty}  \dd r    \\[2mm]
  &\, \leqslant \, \limsup_{n\to\infty}
     \frac{\| \varphi*\mu \|^{}_{\infty}}{\vol (B_n)} \int_{B^{}_{n}}
    \bigl( \lvert \psi \rvert * \lvert \delta^{}_{\!\vL}
    \nts - \nu \rvert \bigr)
            (r)  \dd r  \\[2mm]
  & \, = \, \limsup_{n\to\infty}
     \frac{\| \varphi*\mu \|^{}_{\infty}}{\vol (B_n)} \int_{\RR^d} \int_{\RR^d}
    \! 1^{}_{\nts B^{}_{n}} (r+s) \dd \ts \lvert \delta^{}_{\!\vL}
    \nts - \nu \rvert (s)
        \, \lvert \psi \rvert (r) \dd r \\[2mm]
  & \, = \, \limsup_{n\to\infty}
     \frac{\| \varphi*\mu \|^{}_{\infty}}{\vol (B_n)} \int_{\RR^d}
        \bigl( \nu - \delta^{}_{\!\vL} \bigr) (B^{}_{n} \nts  - r)
        \, \lvert \psi \rvert (r) \dd r \\[2mm]
  & \, \leqslant \, \| \varphi * \mu \|^{}_{\infty} \int_{\RR^d}
        \lvert \psi \rvert (r) \dd r
        \; \limsup_{n\to\infty}
        \frac{\sup_{t \in \RR^d}  \bigl( \nu -  \delta^{}_{\!\vL} \bigr)
            (B_n \nts - t)}{\vol (B_n)}   \ts ,
\end{align*}
where, in the second-last line, we have used the fact that
$\nu - \delta^{}_{\!\vL}$ is a positive measure. The crucial
observation now is that, due to the model set structure with its
uniform distribution properties \cite{Bob-uniform}, the last term
satisfies
\[
     \limsup_{n\to\infty}
        \frac{\sup_{t \in \RR^d}  \bigl( \nu -  \delta^{}_{\!\vL} \bigr)
            (B_n \nts - t)}{\vol (B_n)}     \, = \,
         \dens(\cL) \int_H \bigl( h(t)-1^{}_W(t) \bigr) \dd t
\]
which implies $T_1 < \frac{\eps}{2}$ via \eqref{eq:eps-1}  as
required. }%(end of allowdisplaybreaks)

It now follows from Eq.~\eqref{EQ1} that
$\left| \bigl(\varphi*\psi * \eta \bigr) (0) \right| < \eps $ holds
for all $\eps > 0$. This implies
\[
   \eta \bigl( (I \nts . \ts \varphi) * (I \nts . \ts \psi ) \bigr)
   \, = \,  \eta \bigl( I \nts . (\varphi * \psi) \bigr)
    \, = \,  \bigl( \varphi*\psi * \eta \bigr) (0) \, = \, 0 \ts ,
\]
where $\bigl( I \nts . g\bigr) (x) \defeq g (-x)$ for any (continuous)
function $g$. Since $\varphi, \psi \in \Cc (\RR^d)$ were arbitrary, we
see that $\eta=0$ holds on the subset
\[
    K^{}_{2} (\RR^d) \, \defeq \, \mathrm{span}
       \big\{ f * g :  f, g \in \Cc (\RR^d) \big\} \ts ,
\]
which is dense in $\Cc (\RR^d)$ by a standard approximate identity
argument. Consequently, $\eta=0$, and no other cluster point can
exist.
\end{proof}

Note that Theorem~\ref{thm:FB-trig} remains true if the symmetry
assumption on $\cA$ is lifted. The proof remains unchanged, except for
replacing $B_n$ by $-B_n$ from Eq.~\eqref{eq:conv-approx} onwards, but
commutativity of $\circledast$ is no longer implied; compare
Remark~\ref{rem:commute}.

\section{Eberlein splitting and decomposition for PV
   inflations}\label{sec:decomp}

 Let $\vL^{}_{1} , \ldots , \vL^{}_{\nts N}$ denote pairwise disjoint
 point sets in $\RR^d$ and consider
 $\vL \defeq \ts \dot\bigcup_{i}\, \vL_i$, which we call a \emph{typed
   point set}.  Let us assume that $\vL$ is a Delone set with nice
 averaging properties, including the types. In particular, given some
 symmetric van Hove averaging sequence $\cA$, we are interested in the
 situation that the \emph{pair correlation measures}
\[
      \gamma^{}_{ij} \, \defeq \,  \widetilde{\delta^{}_{\!\vL_i}}
        \nts \circledast \ts \delta^{}_{\!\vL_j}
\]
with respect to $\cA$ exist for all $1 \leqslant i,j \leqslant N$,
where $\widetilde{\mu}$ denotes the (possibly complex) Radon measure
defined by $\widetilde{\mu} (g) = \overline{\mu (\widetilde{g}\ts ) }$
with $g\in \Cc (\RR^d)$ and $\widetilde{g} (x) = \overline{ g(
  -x)}$. For instance, this property is guaranteed when the $\vL_i$
are regular model sets in the same CPS, or when they emerge as the
control points of a primitive inflation rule with $N$ prototiles; see
\cite{TAO1,BG} for background and various details, and \cite{MR} for
extensions. Below, we mainly consider the case $d=1$, though the
setting is general enough to cover higher dimensions as well.

\begin{theorem}\label{thm:decomp}
  Let\/ $\vL= \ts\dot\bigcup_{1\leqslant i \leqslant N}\, \vL_i$ be a
  typed point set generated from of a primitive PV inflation rule in
  one dimension that is aperiodic and has a PV unit as inflation
  factor, and consider the corresponding natural CPS\/
  $(\RR, \RR^m , \cL)$ that emerges via the classic Minkowski
  embedding{\ts}\footnote{The underlying construction is explained in
    full detail in \cite[Sec.~3.4]{TAO1}.}  of the module spanned by
  the points.  Let\/ $W_i$ be the attractors of the induced,
  contractive iterated function system for the windows in internal
  space, and set
\[
  \alpha^{}_i \, \defeq \, \frac{\dens(\vL^{}_{\ts i})}{\dens (\cL)
    \vol(W^{}_{\! i})} \, , \quad
  \omega^{}_i \, \defeq \, \alpha^{}_i \, \delta^{}_{\!\smoplam(W_i)}
    \, , \quad \text{and} \quad
    \nu^{}_i \, \defeq \, \delta^{}_{\!\vL_i}  \! - \ts \omega^{}_i \ts .
\]

Then, for any symmetric van Hove averaging sequence\/ $\cA$ and all\/
$1 \leqslant i,j \leqslant N$, one has the splitting\/
$\delta^{}_{\!\vL_i} = \omega^{}_i \nts + \nu^{}_i$ together with the
decomposition and orthogonality relations
\begin{align*}
   \bigl( \gamma^{}_{ij} \bigr)_{\mathsf{s}} \, = \,
       \widetilde{\omega^{}_i} \circledast {\omega^{}_j}  \, , \quad
   \bigl( \gamma^{}_{ij} \bigr)_{0} \, = \, \widetilde{\nu^{}_i}
    \circledast  {\nu^{}_j} \, , \quad \text{and}
    \quad \widetilde{\omega^{}_i} \circledast {\nu^{}_j} \, =
  \, \widetilde{\nu^{}_{i}} \circledast  \omega^{}_j
  \: = \, 0 \ts ,
\end{align*}
where\/ $\gamma^{}_{ij} =  \bigl( \gamma^{}_{ij} \bigr)_{\mathsf{s}}
+  \bigl( \gamma^{}_{ij} \bigr)_{\mathsf{0}}$ is the unique Eberlein
decomposition of the pair correlation measures into their strongly
almost periodic and their null-weakly almost periodic components.
\end{theorem}

\begin{proof}
  When the inflation factor is a unit, the induced CPS has a Euclidean
  internal space, that is, we have $H = \RR^m$ with $m\geqslant 1$,
  where the latter is a consequence of the assumed aperiodicity. This
  is the situation fully analysed in \cite{BG}, with $\dens \bigl(
  \oplam (W_i) \bigr) \nts = \dens (\cL) \vol (W_i)$.

  Let $\cA$ be arbitrary, but fixed. First, by \cite[Thm.~5.3]{BG},
  the FB coefficients of $\nu^{}_{i}$ satisfy
\[
   c^{}_{\nu_i} (k) \, = \, 0
\]
for all $k \in \RR$. Then, by Theorem~\ref{thm:FB-trig}, we get
$ \widetilde{\omega^{}_i} \circledast {\nu^{}_j} =
\widetilde{\nu^{}_i} \circledast \ts {\omega^{}_j } = 0$ as claimed.

Next, since $\oplam(W^{}_{\! i})$ is a regular model set for each $i$,
in the same CPS, the Eberlein convolutions
$\widetilde{\omega^{}_i} \circledast \ts {\omega^{}_j}$ along $\cA$
exist and are strongly almost periodic measures, that is,
\[
  \widetilde{\omega^{}_i} \circledast {\omega^{}_j}
  \, \in \, \SAP(\RR) \ts .
\]
Likewise, as the $\vL_i$ emerge from a primitive inflation rule, the
pair correlation measures $\gamma^{}_{ij}$ exist and are weakly almost
periodic, compare \cite{BGM,LS2}, so
\[
      \gamma^{}_{ij} \, \in \, \WAP (\RR) \ts .
\]
Therefore, as all required limits exist, we obtain the pair
correlation measures \cite{BGM} as
\begin{align}
    \gamma^{}_{ij} \, & = \, \widetilde{\delta^{}_{\!\vL_i}} \circledast
    {\delta^{}_{\!\vL_j}} \, = \, \bigl(\widetilde{\omega^{}_i} +
     \ts \widetilde{\nu^{}_i}  \bigr) \circledast
    \bigl( {\omega^{}_j} + \ts {\nu^{}_j} \bigr) \nonumber \\[1mm]
   & = \, \widetilde{\omega^{}_i} \circledast {\omega^{}_j} \, + \,
      \widetilde{\omega^{}_i} \circledast {\nu^{}_j}  \, + \,
      \widetilde{\nu^{}_i} \circledast {\omega^{}_j} \, + \,
      \widetilde{\nu^{}_i} \circledast {\nu^{}_j} \nonumber \\[1mm]
   & = \, \widetilde{\omega^{}_i} \circledast {\omega^{}_j} \, + \,
      \widetilde{\nu^{}_i} \circledast {\nu^{}_j} \ts , \label{EQ2}
\end{align}
which implies that every $\widetilde{\nu^{}_{i}} \circledast \nu^{}_{j}$
is a weakly almost periodic measure. Moreover, all these measures are
Fourier transformable by \cite[Lemma~2.1]{BGM}, and we obtain
\[
   \widehat{\gamma^{}_{ij}} \, = \, \widehat{\widetilde{\omega^{}_i}
     \circledast {\omega^{}_j}} \, + \,
   \widehat{\widetilde{\nu^{}_i} \circledast
     {\nu^{}_j}} \ts .
\]

Our system satisfies the \emph{consistent phase property}, which
relates the (generalised) intensities with the Fourier--Bohr
coefficients of the system. This connection is also known as the
Bombieri--Taylor conjecture, and was proved in \cite{Lenz} for a
specific class of one-component systems and later generalised in
\cite{BGM} to the primitive inflation systems under consideration
here. For measures which are pure point diffractive, the consistent
phase property is equivalent to the Besicovitch almost periodicity of
the measure \cite[Thm.~3.36]{LSS2}.  Invoking the results from
\cite{BG}, a simple calculation now gives
\[
    \widehat{\gamma^{}_{ij}}(\{ k \}) \, = \,
     \overline{A^{}_{\!\vL_i} (k)} \ts {A^{}_{\!\vL_j}(k) }
     \, = \,  \widehat{\widetilde{\omega^{}_i}
    \circledast {\omega^{}_j}} ~ ( \{k \})
\]
for all $k\in\RR$, where the amplitudes satisfy
$A^{}_{\! \vL_j} (k) = c^{}_{k} (\delta^{}_{\! \vL_j})$.
But this implies
\[
    \bigl(  \widehat{ \widetilde{\nu^{}_i}  \circledast
       {\nu^{}_j}} \bigr)_{\mathsf{pp}} \, = \, 0 \ts ,
\]
which means that
$\widetilde{\nu^{}_i} \circledast {\nu^{}_j} \in \WAP^{}_0 (\RR)$.
Since
$\widetilde{\omega^{}_i} \circledast {\omega^{}_j} \in \SAP (\RR)$ and
$\widetilde{\nu^{}_i} \circledast {\nu^{}_j} \in \WAP^{}_0 (\RR)$, the
claim follows from Eq.~\eqref{EQ2} and the uniqueness of the Eberlein
decomposition \cite{MoSt}.
\end{proof}

Once again, the result of Theorem~\ref{thm:decomp} remains true
without the symmetry requirement for $\cA$. Indeed, observe that
\[
  \widetilde{ \widetilde{\omega^{}_{i}} \circledast \nu^{}_{\nts j} } \, = \,
  \widetilde{\nu^{}_{\nts j}} \circledast \ts \omega^{}_{i} \ts ,
\]
which is an easy consequence of
$\widetilde{\mu |^{}_{K}} = \widetilde{\mu} |^{}_{-K}$. Then,
Theorem~\ref{thm:FB-trig} still provides the two relations we need to
derive \eqref{EQ2}, though $\circledast$ is no longer implied to be
commutative.

\begin{remark}
  The orthogonality relations in Theorem~\ref{thm:decomp}, via inserting
  the definitions of the measures, also imply that
\[
    \alpha^{}_{i} \, \widetilde{\delta^{}_{\!\smoplam (W_i)}} \circledast
    \delta^{}_{\! \vL_j} \, = \; \alpha^{}_{j} \,
    \widetilde{\delta^{}_{\! \vL_i}} \! \circledast
    \delta^{}_{\!\smoplam (W_j)}
\]
  holds for all $1 \leqslant i,j \leqslant N$, where  $\alpha^{}_{i} =
  \dens (\vL_i) / \! \dens \bigl( \oplam (W_{\nts j})\bigr)$. This
  proportionality of measures matches nicely with the density formula
\begin{equation}\label{eq:polar}
     \widehat{\widetilde{\delta^{}_{\nts P}} \circledast
     \delta^{}_{Q}} ~ \bigl( \{ 0 \} \bigr) \, = \; \dens (P) \, \dens (Q)
\end{equation}
for Meyer sets $P$ and $Q$ that define uniquely ergodic Delone
dynamical systems. With $\mu = \delta^{}_{\nts P}$ and
$\nu = \delta^{}_{Q}$, this relation follows from the complex
polarisation identity
\[
    \widetilde{\mu} \circledast \nu \, = \, \myfrac{1}{4} \sum_{\ell=1}^{4}
     (-\ii)^{\ell} \bigl( \widetilde{\mu + \ii^{\ell} \nu} \bigr)
     \circledast \bigl( \mu + \ii^{\ell} \nu \bigr)  ,
\]
  where the right-hand side is a complex linear combination of four
  positive-definite measures. Then, under Fourier transform, one gets
\[
    \widehat{\widetilde{\mu} \circledast \nu}~ \bigl( \{ 0 \} \bigr) \, = \,
    \myfrac{1}{4} \sum_{\ell=1}^{4} \big| M (\mu) + \ii^{\ell} M ( \nu) \big|^2
\]
from the known intensity formula at $0$, where $M (.)$ is the
\emph{mean} of a measure; see \cite[Prop~9.2]{TAO1} and
\cite[Lemma~4.10.7]{MoSt}. Since $M (\mu) = \dens (P)$ and
$M (\nu) = \dens (Q)$, Eq.~\eqref{eq:polar} follows by expanding the
last sum. It seems interesting to further analyse these identities
under the geometric and topological constraints imposed by the
projection setting.  \exend
\end{remark}

The formulation of Theorem~\ref{thm:decomp} refers to PV inflations
with an inflation multiplier that is a \emph{unit}, because only this
case has been fully treated so far \cite{BG}. It is clear that the key
result, namely \cite[Thm.~5.3]{BG}, can be generalised to the non-unit
situation, and also to PV inflations in higher dimensions. Then, the
decomposition remains the same. Let us illustrate the non-unit
  case with our second guiding example as follows.

\begin{example}\label{ex:TM}
  The Thue{\ts}--Morse (TM) substitution
\[
    \varrho \colon \;  a \mapsto ab \ts , \quad b \mapsto ba
\]
has constant length $2$, and gives rise to a partition of the
integers, $\vL = \ZZ = \vL_a \ts \dot{\cup} \ts \vL_b$. Neither
$\vL_a$ nor $\vL_b$ is a model set, but our above approach is fully
applicable. Here, for the construction of the splitting, we need a CPS
with the $2$-adic integers as internal space, equipped with its
normalised Haar measure that is to be used for the window volumes.
Then, the analogue of the construction from Theorem~\ref{thm:decomp},
for instance with $\cA = \bigl( [-n,n]\bigr)_{n\in\NN}$, leads to
$\omega^{}_{a} = \omega^{}_{b} = \frac{1}{2} \ts \delta^{}_{\ZZ}$ and
thus to the signed measures
\[
  \nu^{}_{a} \, = \, \delta^{}_{\! \vL_a} \! -
  \myfrac{1}{2} \, \delta^{}_{\ZZ}
   \quad \text{and} \quad
   \nu^{}_{b} \, = \, - \nu^{}_{a} \ts .
\]
For $\alpha,\beta \in \{ a,b \}$, the corresponding pair correlation
measures are
\begin{equation}\label{eq:TM}
   \gamma^{}_{\alpha \beta} \, = \, \myfrac{1}{4} \, \delta^{}_{\ZZ}
   \ts + \ts \ts \widetilde{\nts\nu^{}_{\alpha}\nts} \circledast
    \nu^{}_{\beta} \, = \,  \myfrac{1}{4} \, \delta^{}_{\ZZ}
   + \myfrac{1}{4} \, \epsilon^{}_{\alpha \beta} \,
   \gamma_{_\mathrm{TM}} \ts ,
\end{equation}
with $\epsilon^{}_{\alpha \beta } = 1$ or $-1$ depending on whether
$\alpha$ equals $\beta$ or not, and with $\gamma_{_\mathrm{TM}}$
denoting the classic autocorrelation measure of the signed TM sequence
due to Mahler and Kakutani; see \cite[Sec.~10.1]{TAO1} and references
therein for details.

Indeed, \eqref{eq:TM} is the required Eberlein decompositon.  Here,
one has $\widehat{\delta^{}_{\ZZ}} = \delta^{}_{\ZZ}$ by the Poisson
summation formula \cite[Prop.~9.4]{TAO1}, see also \cite{RiStru},
while $\widehat{\gamma_{_\mathrm{TM}}}$ is a positive measure that is
purely singular continuous. It has the Riesz product representation
\[
    \widehat{\gamma_{_\mathrm{TM}}} \, = \prod_{\ell=0}^{\infty}
     \bigl( 1 - \cos \bigl(2^{\ell+1} \pi (.) \bigr) \bigr) ,
\]
to be interpreted as the vague limit of a sequence of absolutely
continuous measures; see \cite{Tanja} for a recent detailed analysis
of this measure.  \exend
\end{example}

The explicit extension of this type of analysis to tiling models in
higher dimensions, for instance in the spirit of \cite{BGM,BFG},
remains an interesting task for the near future, as does a further
splitting of the null-weakly almost periodic part, as in \cite{Nicu2},
according to the different continuous types after Fourier transform.

\section{Further directions}\label{sec:further}

It is natural to ask whether our constructive approach can be made to
work also beyond the PV inflation case, and thus perhaps add extra
insight into systems that are covered by \cite{Au}. In particular, it
is of interest to see the orthogonality in the Eberlein sense in more
generality, at least almost surely with respect to some given
invariant measure in a dynamical systems context. This is indeed
possible, and we outline this with two concrete cases.

\subsection{Eberlein splitting for a lattice gas}

Let $0<p<1$ be fixed, and consider the Bernoulli \emph{lattice gas} on
$\ZZ$ with independent occupation probability $p$ for the single
sites. A realisation of this process can thus either be seen as a
configuration, meaning an element of $\{ 0 ,1 \}^{\ZZ}$, or as a
subset $\vL \subseteq \ZZ$.  Let us take the latter view, and select a
typical subset $\vL$, thus with density $p$. By \cite[Ex.~11.2]{TAO1},
the corresponding autocorrelation measure is
\[
    \gamma \, = \, p^2 \ts \delta^{}_{\ZZ} + p (1-p)\ts \delta^{}_{0}
    \, = \, (\gamma)^{}_{\mathsf{s}} + (\gamma)^{}_{0} \ts ,
\]
which applies to almost all realisations of the Bernoulli process.
This can easily be proved with the strong law of large numbers, see
\cite{BM-ran} for a detailed exposition, and provides the Eberlein
decomposition of $\gamma$ into its strongly almost periodic and its
null-weakly almost periodic part. Indeed, the diffraction measure then
simply is
$\widehat{\gamma} = p^2 \ts \delta^{}_{\ZZ} + p (1-p)
\lambda^{}_{\mathrm{L}}$, where $\lambda^{}_{\mathrm{L}}$ denotes
Lebesgue measure and $\widehat{\gamma}$ applies to almost all
realisations.

Now, in analogy to our above splitting, we set
\begin{equation}\label{eq:Bern-split}
     \delta^{}_{\! \vL} \, = \, \omega + \nu
     \qquad \text{with }
     \; \omega \, = \, p \, \delta^{}_{\ZZ}
     \; \text{ and } \; \nu \, = \, \delta^{}_{\!\vL} - \omega \ts .
\end{equation}
Using $\cA = \bigl( [-n,n]\bigr)_{n\in\NN}$ as before, we get
$\delta^{}_{\ZZ} \circledast \delta^{}_{\ZZ} = \delta^{}_{\ZZ}$ from
\cite[Ex.~8.10]{TAO1}, hence
\[
    \omega \circledast \widetilde{\omega} \, = \,
    \omega \circledast \omega \, = \, p^2 \ts \delta^{}_{\ZZ}
    \, = \, (\gamma)^{}_{\mathsf{s}}
\]
and $\omega \circledast \delta^{}_{\ZZ} = p \,
\delta^{}_{\ZZ}$. Furthermore,
$\delta^{}_{\!\vL} \circledast \delta^{}_{\ZZ} = \delta^{}_{-\vL}
\circledast \delta^{}_{\ZZ} = p \ts\ts \delta^{}_{\ZZ}$ follows from a
simple density calculation, where one observes that $-\vL$ is another
typical realisation if $\vL$ is one. But this implies the
orthogonality relations
$\ts\omega \circledast \widetilde{\nu} = \widetilde{\omega}
\circledast \nu = 0$ together with
\[
    \nu \circledast \widetilde{\nu} \, = \,
    \delta^{}_{\!\vL} \circledast \delta^{}_{-\vL} -
    p^2 \ts \delta^{}_{\ZZ} \, = \, \gamma - p^2 \ts \delta^{}_{\ZZ}
    \, = \, p (1-p) \, \delta^{}_{0} \, = \, ( \gamma )^{}_{0} \ts .
\]
We thus see that \eqref{eq:Bern-split} provides the splitting of
$\delta^{}_{\!\vL}$ in complete analogy to Theorem~\ref{thm:decomp},
and applies to almost all realisations of the lattice gas.

Clearly, this can be extended to higher dimensions, for instance via
replacing $\ZZ$ by $\ZZ^d$ in the above example, and to many other
stochastic systems, as treated in \cite[Ch.~11]{TAO1} or in
\cite{BBM}.  The constructive splitting approach gives a slightly
different interpretation to the method put forward in \cite{Luck,BBM},
where the splitting is done on the autocorrelation level by separating
the mean from the fluctuations.  This certainly deserves further
clarification in the setting of the Bartlett spectrum from the theory
of stochastic processes.

\subsection{Eberlein splitting for a random inflation}

Here, we return to the classic Fibonacci inflation that underlies
Section~\ref{sec:Fibo} and turn it into a \emph{random inflation}
by setting
\[
    \varrho \colon \;   b \mapsto a \, , \quad a \mapsto
    \begin{cases} a b , & \text{with probability } p , \\
        b a , & \text{with probability } 1-p ,  \end{cases}
\]
where $p \in [0,1 ]$ is fixed, and the rule is applied
\emph{locally}. This defines a system that was first analysed in
\cite{GL}, and has recently turned into an interesting paradigm for a
random system with both long-range order and some form of disorder
\cite{RS,GRS,GS}. Here, we only consider the geometric setting where
$a$ and $b$ stand for intervals of length $\tau$ and $1$,
respectively, with normalised Dirac measures on their left endpoints.

From \cite[Thm.~3.19]{BSS}, we know that the diffraction measure
almost surely satisfies
\begin{equation}\label{eq:ranfib-1}
    \widehat{\gamma} \, = \, \bigl( \widehat{\gamma} \bigr)_{\mathsf{pp}}
    + \bigl( \widehat{\gamma} \bigr)_{\mathsf{ac}} \ts ,
\end{equation}
where explicit expressions can be given. Clearly,
$( \widehat{\gamma} )^{}_{\mathsf{ac}} = 0$ for $p=0$ or $p=1$.  To
set this into our above scheme, let
$\vL = \vL_{a} \,\dot\cup\, \vL_{b}$ be the (typed) control points of
a typical realisation of the random Fibonacci inflation, say with
$0<p<1$ to avoid the deterministic limiting cases. Then, as shown in
\cite{BSS}, there are weighted Dirac combs $\omega^{}_{\alpha}$, with
$\alpha\in \{ a , b\}$, of the form
\[
    \omega^{}_{\alpha} \, = \sum_{x \in \ZZ[\tau]}
    h^{}_{\alpha} (x^{\star}) \, \delta^{}_{x}
\]
with functions $h^{}_{\alpha}$ that are continuous and supported on
the window $[-\tau,\tau]$ of the covering model set. In particular,
these Dirac combs are discretely supported in the regular model set
$\oplam \bigl( [-\tau, \tau] \bigr)$, and both have pure point
spectrum \cite{NS11}.  What is more, as follows from \cite{BSS}, the
pure point part of the diffraction of $\vL_{\alpha}$, agrees with the
diffraction of $\omega^{}_{\alpha}$.

Thus, using $\omega^{}_{\alpha}$ and defining
$\nu^{}_{\nts \alpha} = \delta^{}_{\! \vL_{\alpha}} \! - \ts
\omega^{}_{\alpha}$, where the latter is a random measure for each
$\alpha \in \{ a, b\}$, we almost surely (in the sense of the
underlying process) are in a situation analogous to the one from
Theorem~\ref{thm:decomp}. In particular, for fixed
$u^{}_{a}, u^{}_{b} \in \CC$, the random weighted Dirac comb
\[
     \mu \, = \, u^{}_{a} \, \delta^{}_{\! \vL_{a}}  \!
                  + \ts u^{}_{b} \, \delta^{}_{\! \vL_{b}}
\]
almost surely has the autocorrelation
$\gamma = (\gamma)^{}_{\mathsf{s}} + (\gamma)^{}_{0}$ where
$(\gamma)^{}_{\mathsf{s}}$ agrees with the autocorrelation of
$u^{}_{a}\ts \omega^{}_{a} + u^{}_{b}\ts \omega^{}_{b}$ and
$(\gamma)^{}_{0}$ is the autocorrelation of
$u^{}_{a} \nu^{}_{\nts a} + u^{}_{b} \nu^{}_{b}$. Taking Fourier
transforms brings us back to \eqref{eq:ranfib-1}, where we refer to
\cite{Moll,Timo,BSS} for explicit formulas. Once again, this gives a
constructive variant of the decomposition advocated in \cite{Au},
which is fully compatible with the statistical separation of mean and
variance from \cite{Luck,BBM}.  In fact, the latter approach quite
generally seems to lead to a related further decomposition into
singular versus absolutely continuous components, at least for the
class of random inflations \cite{Gohlke}, which opens a promising path
to future investigations.

\section*{Acknowledgements}

We thank Philipp Gohlke, Uwe Grimm, Neil Ma\~{n}ibo and Timo Spindeler
for helpful discussions and comments on the manuscript.  We thank two
referees for their thoughtful comments, which helped to improve the
presentation.  This work was supported by the German Research
Foundation (DFG), within the CRC 1283 at Bielefeld University (MB),
and by the Natural Sciences and Engineering Council of Canada (NSERC),
via grant {2020-00038} (NS).

\end{document}